\documentclass[11pt]{article}
\usepackage{color,soul}
\usepackage{titlesec}
\usepackage{breakcites}
\usepackage[pagebackref]{hyperref}
\usepackage{dirtytalk}
\usepackage{lipsum}
\usepackage{ntheorem}
\usepackage{mathtools}
\usepackage{graphicx}
\usepackage{calc}
\usepackage{stmaryrd}
\usepackage{paralist}
\newlength{\depthofsumsign}
\usepackage{amsmath,amssymb,amsfonts}
\usepackage{booktabs}
\usepackage{array}

\usepackage{tikz}
\usepackage{hhline}
\usepackage{multirow}
\usepackage{subfig}
\usepackage{hyperref}
\usepackage{mathtools}
\usepackage{tikz}
\usepackage{makecell}
\usepackage{pgfplots}
\usepackage{afterpage}
\usepackage{fancyhdr}
\usepackage{setspace}
\usepackage{bm}
\usepackage[linesnumbered,ruled,vlined]{algorithm2e}
\SetKwInput{KwInput}{Input}
\SetKwInput{KwOutput}{Output}
\usepackage{authblk}

\DeclarePairedDelimiter\floor{\lfloor}{\rfloor}

\newtheorem{corollary}{Corollary}
\newtheorem{proposition}{Proposition}
\newtheorem*{proposition-non}{Proposition}

\newtheorem{theorem}{Theorem}

\newtheorem{remark}{Remark}

\newtheorem{construction}{Construction}
\newtheorem{method}{Method}

\allowdisplaybreaks

\newenvironment{proof}{\noindent{\bf Proof:}\indent}%
                      {\hfill $\Box$\par}

\newcommand{\sym}[1]{{\sf #1}}

\title{Constructions of Efficiently Implementable Boolean Functions with Provable Nonlinearity/Resiliency/Algebraic Immunity Trade-Offs}
\author[ ]{Palash Sarkar}
\affil[ ]{Indian Statistical Institute, 203, B.T. Road, Kolkata, India 700108}
\affil[ ]{Email: {\tt palash@isical.ac.in}}

\date{\today}

\begin{document}

\maketitle

\begin{abstract}
We describe new families of efficiently implementable Boolean functions achieving provable trade-offs between resiliency, nonlinearity, and 
algebraic immunity. Compared to one direction of research on balanced and 1-resilient functions with optimal algebraic immunity, the
new functions provide lower algebraic immunity, but much lower -- in fact, almost optimal -- linear bias and importantly are much more efficient to implement. 
Compared to another direction of research based on constructions using direct sum, the new functions while achieving similar efficiency, has less number
of variables for the same order of resiliency, same lower bound on algebraic immunity, and similar nonlinearity.

\noindent{\bf Keywords:} Boolean function, resiliency, nonlinearity, algebraic immunity, efficient implementation.
\end{abstract}

\section{Introduction \label{sec-intro}}
Boolean functions have widespread applications in various areas of computer science and engineering. For cryptographic applications, especially in
the construction of stream ciphers, certain properties of 
Boolean functions have been identified as necessary for providing resistance to known attacks. Three such extensively studied properties are resiliency, 
nonlinearity and
algebraic immunity. Over the last few decades, extensive research has been carried out on various aspects of Boolean functions possessing one or more of these three
properties. We refer to the excellent book~\cite{BF-book} for a very comprehensive and unified treatment of cryptographic properties of Boolean functions.

It is easy to construct functions which maximise any one of the three properties of resiliency, nonlinearity and algebraic immunity. For an $n$-variable function,
the maximum possible order of resiliency is $n-1$ and the only functions which achieve this order of resiliency are the two affine functions which are non-degenerate
on all the $n$ variables. For even $n$, the maximum possible nonlinearity is achieved by bent functions~\cite{DBLP:journals/jct/Rothaus76} and there are many well known constructions
of bent functions. The maximum possible algebraic immunity~\cite{DBLP:conf/eurocrypt/CourtoisM03} of an $n$-variable function is $\lceil n/2\rceil$, and the majority function 
achieves this value of algebraic immunity~\cite{DBLP:journals/dcc/DalaiMS06}. 
Functions that maximise one property usually perform poorly with respect to the other two.
Affine functions have minimum nonlinearity and algebraic immunity, bent functions are not balanced (i.e. not 0-resilient), while
the majority function has poor nonlinearity and resiliency. 

While resiliency, nonlinearity and algebraic immunity are security properties of Boolean functions, there is another aspect of Boolean functions which is also of great
practical importance. For use in actual design of cryptographic systems, it is crucial that any Boolean function used in the design is efficient to implement. 
A measure of implementation
efficiency is the number of 2-input gates that is required to implement a Boolean function. From a cryptographic point of view, along with good security properties, a 
Boolean function also needs to have a low gate count. 

Previous works have proceeded in two directions. The first direction is focused on constructing functions which achieve optimal algebraic immunity,
while at the same time achieving guarantees on the nonlinearity of the constructed functions. For balanced functions, this line of work has been reported 
in~\cite{DBLP:conf/asiacrypt/CarletF08,DBLP:journals/tit/WangPKX10,DBLP:journals/tit/Carlet11,DBLP:journals/tit/TangCT13,DBLP:journals/jossac/ShanHZL18,DBLP:journals/dcc/WangS19,DBLP:journals/dam/Hu0H20}, and for 1-resilient functions on an even number of variables, this line of work has been reported 
in~\cite{DBLP:journals/dam/TuD12,DBLP:conf/cisc/WangLL12,DBLP:journals/ijfcs/TangCT14,DBLP:journals/ijcm/WangZWZW16,DBLP:journals/tit/TangCTZ17}.
These works do not consider efficiency of implementation to be an issue and 
the constructed functions essentially require exponential (in the number of variables) size circuits to implement.
Also, to the best of our knowledge, work in this direction has not considered order of resiliency greater than 1. More generally, the motivation
for this line of research is essentially theoretical, and progress is difficult and slow.

The second independent direction of work~\cite{DBLP:conf/eurocrypt/MeauxJSC16,DBLP:conf/indocrypt/MeauxCJS19,DBLP:journals/tit/CarletM22,DBLP:conf/africacrypt/GerardGM24} has proposed functions which are constructed as a direct sum of different kinds of functions,
In particular, the construction in~\cite{DBLP:conf/africacrypt/GerardGM24} (which is also in the long version of~\cite{DBLP:conf/indocrypt/MeauxCJS19})
is the direct sum of a linear function, a bent function, and a threshold function, which provides 
separate control over the three parameters of resiliency, nonlinearity, and algebraic immunity.
The functions constructed by the methods in~\cite{DBLP:conf/eurocrypt/MeauxJSC16,DBLP:conf/indocrypt/MeauxCJS19,DBLP:journals/tit/CarletM22,DBLP:conf/africacrypt/GerardGM24} are efficient to implement. These functions are targeted towards the application of homomorphic encryption.

In this paper, we contribute to the second direction of research. A recent work~\cite{DBLP:journals/tit/CarletS26} proposed the construction of efficiently
implementable bent functions with provable lower bound on algebraic immunity. We leverage this result to provide constructions of functions which are
efficient to implement and provide provable trade-offs between resiliency, nonlinearity, and algebraic immunity. 

Two of our constructions are basic. The first construction increases resiliency by simply adding new variables, which provides a separate control
over resiliency to the bent functions from~\cite{DBLP:journals/tit/CarletS26}. For the resulting functions, 
the number of variables and the order of resiliency do not have the same parity. A second construction of a similar type shows how to increase resiliency
by first adding a 5-variable, 1-resilient function to the bent function from~\cite{DBLP:journals/tit/CarletS26}, and then adding an appropriate number of 
additional variables. For the second construction, the
number of variables and the order of resiliency have the same parity. Both the constructions provide functions which are essentially as efficient
to implement as the bent functions from~\cite{DBLP:journals/tit/CarletS26}. We provide a detailed comparison of the 
first construction to the approach in~\cite{DBLP:conf/indocrypt/MeauxCJS19,DBLP:conf/africacrypt/GerardGM24} where the threshold function
of~\cite{DBLP:conf/indocrypt/MeauxCJS19,DBLP:conf/africacrypt/GerardGM24} is instantiated with the majority function.
We also show that for balanced functions, compared to~\cite{DBLP:conf/asiacrypt/CarletF08,DBLP:journals/tit/WangPKX10,DBLP:journals/tit/Carlet11,DBLP:journals/tit/TangCT13,DBLP:journals/jossac/ShanHZL18,DBLP:journals/dcc/WangS19,DBLP:journals/dam/Hu0H20}, the new functions provide same algebraic immunity,
much lower linear bias, and are more efficient to implement even though the new functions require more variables. A similar comparison
holds for the approach to construct 1-resilient functions taken in~\cite{DBLP:journals/dam/TuD12,DBLP:conf/cisc/WangLL12,DBLP:journals/ijfcs/TangCT14,DBLP:journals/ijcm/WangZWZW16,DBLP:journals/tit/TangCTZ17}.

Both the above constructions achieve resiliency by essentially adding variables. This creates simple linear structures, which may be considered undesirable.
We provide two iterated constructions by incorporating the bent function constructed in~\cite{DBLP:journals/tit/CarletS26} into a two-decade
old construction idea sketched in~\cite{PASALIC2001158}. The resulting constructions provide functions which are also efficient to implement, and 
which provide resiliency/nonlinearity/algebraic immunity trade-offs that are different from the first two constructions.
Compared to the basic constructions, the iterated constructions yield functions, which for the same order of resiliency, and the same lower bound
on algebraic immunity, provide lower linear bias at the cost of a few extra variables and 2-input gates.

From the viewpoint of theoretical computer science, the properties of algebraic immunity and resiliency were shown~\cite{DBLP:journals/siamcomp/ApplebaumL18} to be key parameters of 
the ``local function'' required in Goldreich's construction~\cite{DBLP:journals/eccc/ECCC-TR00-090} of pseudorandom generators. Motivated 
by~\cite{DBLP:journals/siamcomp/ApplebaumL18}, a recent study~\cite{DBLP:journals/dcc/DupinMR23} dealt quite extensively with the provable trade-off between resiliency and algebraic 
immunity. This line of work does not consider nonlinearity of the functions. Since we consider nonlinearity as a parameter of the trade-off, the
functions constructed in~\cite{DBLP:journals/dcc/DupinMR23} are not directly comparable to our work.

\subsection*{Outline of the paper}
The background and preliminaries are described in Section~\ref{sec-prelim}. In Section~\ref{sec-basic}, we present the two basic constructions. The iterated
constructions are presented in Section~\ref{sec-iterate}. The four constructions are compared in Section~\ref{sec-comp}. Finally, Section~\ref{sec-conclu}
concludes the paper.

\section{Background and Preliminary Results \label{sec-prelim}}
In this section, we provide the basic definitions and present some preliminary results. For further and extensive details on cryptographic properties of
Boolean functions we refer to~\cite{BF-book}.

The cardinality of a finite set $S$ will be denoted by $\#S$.
$\mathbb{F}_2$ denotes the finite field of two elements; $\mathbb{F}_2^n$, where $n$ is a positive integer, denotes the vector space of dimension $n$ over $\mathbb{F}_2$. 
The addition operator over both $\mathbb{F}_2$ and $\mathbb{F}_2^n$ will be denoted by $\oplus$. The product (which is also the logical AND) of $x,y\in\mathbb{F}_2$
will simply be written as $xy$.
A bit vector of dimension $n$, i.e. an element of $\mathbb{F}_2^n$ will also be considered to be an $n$-bit binary string.

Let $n$ be a positive integer. The support of $\mathbf{x}=(x_1,\ldots,x_n)\in \mathbb{F}_2^n$ is $\sym{supp}(\mathbf{x})=\{1\leq i\leq n: x_i=1\}$,
and the weight of $\mathbf{x}$ is $\sym{wt}(\mathbf{x})=\#\sym{supp}(\mathbf{x})$. 
By $\mathbf{0}_n$ and $\mathbf{1}_n$ we will denote the all-zero and all-one strings of length $n$ respectively. 
For $\mathbf{x},\mathbf{y}\in\mathbb{F}_2^n$, with $\mathbf{x}=(x_1,\ldots,x_n)$ and $\mathbf{y}=(y_1,\ldots,y_n)$ the distance between $\mathbf{x}$ and $\mathbf{y}$
is $d(\mathbf{x},\mathbf{y})=\#\{i: x_i\neq y_i\}$; the inner product of $\mathbf{x}$ and $\mathbf{y}$ is 
$\langle \mathbf{x},\mathbf{y}\rangle = x_1y_1 \oplus \cdots \oplus x_ny_n$. 

\paragraph{Boolean function.}
An $n$-variable Boolean function $f$ is a map $f:\mathbb{F}_2^n\rightarrow \mathbb{F}_2$. 
The weight of $f$ is $\sym{wt}(f)=\#\{\mathbf{x}\in\mathbb{F}_2^n:f(\mathbf{x})=1\}$; $f$ is said to be \textit{balanced} if $\sym{wt}(f)=2^{n-1}$. 
The canonical ordering of the elements of $\mathbb{F}_2^n$ is the ordering where for $0\leq i<2^n$, the $i$-th element in the ordering is the $n$-bit binary representation 
of $i$. With respect to the canonical ordering an $n$-variable function $f$ can be represented by a 
binary string of length $2^n$, where the $i$-th bit of the string is the value of $f$ on the $i$-th element of the canonical representation. 
We call such a string to be the \textit{string (or truth table) representation} of $f$.

The \textit{algebraic normal form (ANF) representation} of an $n$-variable Boolean function $f$ is the representation of $f$ as an element of the 
polynomial ring $\mathbb{F}_2[X_1,\ldots,X_n]/(X_1^2\oplus X_1,\ldots,X_n^2\oplus X_n)$ in the following manner:
$f(X_1,\ldots,X_n) = \bigoplus_{\bm{\alpha}\in\mathbb{F}_2^n} a_{\bm{\alpha}} \mathbf{X}^{\bm{\alpha}}$, where 
$\mathbf{X}=(X_1,\ldots,X_n)$; for $\bm{\alpha}=(\alpha_1,\ldots,\alpha_n)\in \mathbb{F}_2^n$, $\mathbf{X}^{\bm{\alpha}}$ denotes the monomial
$X_1^{\alpha_1}\cdots X_n^{\alpha_n}$; and $a_{\bm{\alpha}} \in \mathbb{F}_2$. 
The (algebraic) degree of $f$ is $\sym{deg}(f)=\max\{\sym{wt}(\bm{\alpha}):a_{\bm{\alpha}}=1\}$; we adopt the convention that the zero function has degree 0.
The degree (or sometimes also called the length) of the monomial $\mathbf{X}^{\bm{\alpha}}$ is $\sym{wt}(\bm{\alpha})$.

\textit{In the following, we will define methods to combine two or more Boolean functions to construct a new Boolean function. Some of these methods
will be the focus of our study and we will call such methods to be constructions, while the others are intermediate methods and we will simply 
refer to them as methods.}

It is useful to introduce a notation for the concatenation of two functions.
	Let $g$ and $h$ be $n$-variable functions. Define an $(n+1)$-variable function in the following manner.
	\begin{eqnarray*}
		f(X_1,\ldots,X_{n+1}) & = & (1\oplus X_{n+1})g(X_1,\ldots,X_n) \oplus X_{n+1}h(X_1,\ldots,X_n).
	\end{eqnarray*}
	We denote $f$ as $\sym{Concat}(g,h)$.
Note that the string representation of $f=\sym{Concat}(g,h)$ is obtained by concatenating the string representations of $g$ and $h$.

An $n$-variable function $f(X_1,\ldots,X_n)$ is said to be \textit{non-degenerate} on the variable $X_i$, $1\leq i\leq n$, if there are $\bm{\alpha},\bm{\beta}\in\mathbb{F}_2^n$ 
which differ only in the $i$-th position and $f(\bm{\alpha})\neq f(\bm{\beta})$; if there are no such $\bm{\alpha}$ and $\bm{\beta}$, then $f$ is said to be
degenerate on the variable $X_i$. 

Functions of degree at most 1 are said to be affine functions. Affine functions with $a_{\mathbf{0}_n}=0$ are said to be linear functions. 
Each $\bm{\alpha}=(\alpha_1,\ldots,\alpha_n)\in\mathbb{F}_2^n$, defines the linear function 
$\langle\bm{\alpha},\mathbf{X} \rangle = \langle \bm{\alpha},(X_1,\ldots,X_n)\rangle=\alpha_1X_1\oplus\cdots\oplus \alpha_nX_n$. If $\sym{wt}(\bm{\alpha})=w$, then the function
$\langle \bm{\alpha},\mathbf{X} \rangle$ is non-degenerate on exactly $w$ of the $n$ variables. 

For $\mathbf{X}=(X_1,\ldots,X_n)$, by $L_n(\mathbf{X})$ we will denote the following function.
\begin{eqnarray}\label{eqn-L}
	L_n(\mathbf{X}) & = & X_1\oplus \cdots X_n.
\end{eqnarray}


%

\paragraph{Nonlinearity and Walsh transform.}
The distance between two $n$-variable functions $f$ and $g$ is $d(f,g)=\#\{\mathbf{x}\in\mathbb{F}_2^n:f(\mathbf{x})\neq g(\mathbf{x})\}=\sym{wt}(f\oplus g)$.
The \textit{nonlinearity} of an $n$-variable function $f$ is defined to be 
$\sym{nl}(f) = \min_{\bm{\alpha}\in\mathbb{F}_2^n} \{d(f,\langle \bm{\alpha},\mathbf{X} \rangle), d(f,1\oplus \langle \bm{\alpha},\mathbf{X} \rangle)\}$, 
i.e. the nonlinearity of $f$ is the minimum of the distances of $f$ to all the affine functions. 

The Walsh transform of an $n$-variable function $f$ is the map $W_f:\mathbb{F}_2^n\rightarrow \mathbb{Z}$, where for $\bm{\alpha}\in\mathbb{F}_2^n$,
$W_f(\bm{\alpha}) = \sum_{\mathbf{x}\in\mathbb{F}_2^n} (-1)^{f(\mathbf{x}) \oplus \langle \bm{\alpha}, \mathbf{x} \rangle}$. 
From the definition it follows that $W_f(\bm{\alpha})=2^n-2d(f,\langle \bm{\alpha},\mathbf{X} \rangle)$. Consequently, 
$\sym{nl}(f) = 2^{n-1} - \frac{1}{2}\max_{\bm{\alpha} \in \mathbb{F}_2^n} |W_f(\bm{\alpha})|$. The nonlinearity of $f$ is invariant under an invertible
linear transformation on the variables of $f$. 

\paragraph{Linear bias.}
We define the \textit{linear bias} of an $n$-variable function $f$ to be 
$\sym{LB}(f)=\max_{\bm{\alpha}\in\mathbb{F}_2^n}\lvert W_f(\bm{\alpha})\rvert/2^n=1-\sym{nl}(f)/2^{n-1}$. From a cryptographic point of view, 
the linear bias, rather than the nonlinearity, is of importance, since it is the linear bias which is used to quantify the resistance to (fast) correlation attacks.

\paragraph{Bent functions.}
An $n$-variable function $f$ is said to be bent~\cite{DBLP:journals/jct/Rothaus76} if $W_f(\bm{\alpha})=\pm 2^{n/2}$ for all $\bm{\alpha}\in\mathbb{F}_2^n$.
From the definition it follows that bent functions can exist only if $n$ is even. An $n$-variable bent function has nonlinearity $2^{n-1} - 2^{n/2-1}$ 
(resp. linear bias $2^{-n/2}$), and this is the maximum possible nonlinearity (resp. least possible linear bias) that can be achieved by any $n$-variable function.

\paragraph{Maiorana-McFarland bent functions.}
The well known Maiorana-McFarland class of bent functions is defined as follows. For $k\geq 1$, let $\psi:\{0,1\}^k\rightarrow\{0,1\}^k$ be a bijection and 
$h:\{0,1\}^k\rightarrow \{0,1\}$ be a Boolean function. 
Let $\mathbf{X}=(X_1,\ldots,X_k)$ and $\mathbf{Y}=(Y_1,\ldots,Y_k)$. For $k\geq 1$, $(\psi,h)\mbox{-}\sym{MM}_{2k}$ is defined to be the following function.
\begin{eqnarray}
	(\psi,h)\mbox{-}\sym{MM}_{2k}(\mathbf{X},\mathbf{Y}) & = & \langle \psi(\mathbf{X}),\mathbf{Y}\rangle \oplus h(\mathbf{X}). \label{eqn-MM-even} 
\end{eqnarray}
Note that the degree of $(\psi,h)\mbox{-}\sym{MM}_{2k}$ is $\max(1+\max_{1\leq i\leq k}\sym{deg}(\psi_i),\sym{deg}(h))$, where 
$\psi_1,\ldots,\psi_k$ are the coordinate functions of $\psi$.

By $Q_{2k}(\mathbf{X},\mathbf{Y})$ we will denote $(\psi,h)\mbox{-}\sym{MM}_{2k}$ with $\psi$ to be the identity permutation, and $h$ to be the constant 
all-zero function, i.e.
\begin{eqnarray}\label{eqn-Q}
	Q_{2k}(\mathbf{X},\mathbf{Y}) & = & X_1Y_1\oplus \cdots X_kY_k.
\end{eqnarray}

\paragraph{Almost optimal linear bias.} 
For a positive integer $n$, the covering radius bound (see Theorem~3 in Chapter~3 of~\cite{BF-book}) on the nonlinearity of an $n$-variable function $f$ is the following:
$\sym{nl}(f)\leq 2^{n-1}-2^{n/2-1}$, and equivalently, $\sym{LB}(f)\geq 2^{-n/2}$. The bound is achieved if and only if $f$ is bent. 
We say that $f$ has \textit{almost optimal linear bias} if $2^{-n/2} \leq \sym{LB}(f)\leq 2^{-n/2+1}$, 
i.e. if the linear bias of $f$ is at most two times the lower bound arising from the covering radius bound.
If $n$ is even, and $\sym{nl}(f)=2^{n-1}-2^{n/2}$, then $\sym{LB}(f)=2^{-n/2+1}$,
while if $n$ is odd, and $\sym{nl}(f)=2^{n-1}-2^{(n-1)/2}$, then $\sym{LB}(f)=2^{-n/2+1/2}<2^{-n/2+1}$; in both cases the linear bias is almost optimal.

\paragraph{Resilient functions.}
Let $n$ be a positive integer and $m$ be an integer such that $0\leq m<n$. 
An $n$-variable function $f$ is said to be $m$-resilient if $W_f(\bm{\alpha})=0$ for all $\bm{\alpha}$ satisfying $\sym{wt}(\bm{\alpha})\leq m$. 
Equivalently, $f$ is $m$-resilient if and only if $d(f,\langle \bm{\alpha},\mathbf{X} \rangle)=2^{n-1}$ for $\bm{\alpha}$ satisfying $\sym{wt}(\bm{\alpha})\leq m$, i.e.
if the distance between $f$ and any linear function which is non-degenerate on at most $m$ variables is equal to $2^{n-1}$.
Siegenthaler's bound~\cite{DBLP:journals/tit/Siegenthaler84} relates $n$, $m$ and the degree $d$ of $f$ in the following manner: 
\begin{eqnarray}\label{eqn-s-bound}
	\mbox{if $m=n-1$, then $d=1$, and if $m\leq n-2$, then $d\leq n-m-1$.} 
\end{eqnarray}
Divisibility results obtained in~\cite{DBLP:conf/crypto/SarkarM00,DBLP:conf/seta/Carlet01,CS02} show that the Walsh transform values of an $n$-variable,
$m$-resilient function $f$ having degree $d$ is divisible by $2^{m+2+\lfloor(n-m-2)/d\rfloor}$. 

\paragraph{Algebraic immunity.}
The algebraic immunity of an $n$-variable function $f$ is defined~\cite{DBLP:conf/eurocrypt/CourtoisM03} as follows:
$\sym{AI}(f)=\min_{g\neq 0} \{\sym{deg}(g): \mbox{ either } gf=0, \mbox{ or }  g(f\oplus 1)=0\}$.
It is known~\cite{DBLP:conf/eurocrypt/CourtoisM03} that $\sym{AI}(f)\leq \lceil n/2\rceil$. (The result follows from a restatement
of Theorem~6.0.1 of~\cite{DBLP:conf/eurocrypt/CourtoisM03} in the extended version of~\cite{DBLP:conf/eurocrypt/CourtoisM03} available
at \url{http://ntcourtois.free.fr/toyolili.pdf}.)
If the bound is achieved, then we say that $f$ has optimal algebraic immunity.

Let $\sym{AN}(f)=\min_{g\neq 0}\{\sym{deg}(g): fg=0\}$, and $\Delta_{\sym{AN}}(f)=\mid \sym{AN}(f)-\sym{AN}(f\oplus 1)\mid$. Note that
$\sym{AI}(f)=\min\{\sym{AN}(f),\sym{AN}(f+1)\}$.

\paragraph{Fast algebraic immunity.}
Given an $n$-variable function $f$, suppose that there are $n$-variable functions $g\neq 0$ and $h$ of degrees $e$ and $d$ respectively such that
$gf=h$. If $e+d\geq n$, then the existence of $g$ and $h$ satisfying $gf=h$ is guaranteed~\cite{DBLP:conf/crypto/Courtois03}. The 
fast algebraic immunity (FAI) of $f$ is defined in the following manner: 
$\sym{FAI}(f)\allowbreak =\allowbreak \min \left( 2\sym{AI}(f), \allowbreak
\min_{g\neq 0}\allowbreak \{\sym{deg}(g) \allowbreak + \allowbreak \sym{deg}(fg):
\allowbreak 1 \allowbreak \leq \allowbreak \sym{deg}(g) \allowbreak < \allowbreak \sym{AI}(f)\}\right)$.
The following bounds hold for $\sym{FAI}(f)$: $1+\sym{AI}(f)\leq \sym{FAI}(f)\leq 2\,\sym{AI}(f)$. In particular, the lower bound
of $1+\sym{AI}(f)$ on $\sym{FAI}(f)$ suggests that a target lower bound (which is not necessarily optimal) on $\sym{FAI}(f)$ may be achieved by designing 
functions with a desired value of algebraic immunity. 

\paragraph{Majority function.}
For $n\geq 1$, let $\sym{Maj}_n:\{0,1\}^n\rightarrow \{0,1\}$ be the majority function defined in the following manner.
For $\mathbf{x}\in\{0,1\}^n$, $\sym{Maj}(\mathbf{x}) = 1$ if and only if $\sym{wt}(\mathbf{x}) > \floor{n/2}$. 
\begin{theorem}[From Theorems~1,~2, and~3 of~\cite{DBLP:journals/dcc/DalaiMS06}] \label{thm-maj-prop}
Let $n$ be a positive integer. 
	\begin{compactenum}
		\item $\sym{Maj}_n$ has the maximum possible AI of $\lceil n/2\rceil$.
		\item The degree of $\sym{Maj}_n$ is equal to $2^{\floor{\log_2n}}$.
		\item $\sym{nl}(\sym{Maj}_n) = 2^{n-1}-{n-1 \choose \lfloor n/2\rfloor}$, 
			equivalently $\sym{LB}(\sym{Maj}_n) = {n-1 \choose \lfloor n/2\rfloor}/2^{n-1}$.
	\end{compactenum}
\end{theorem}
\begin{proposition}[Theorem~4.1 of~\cite{We87}]\label{prop-maj-gc}
	$\sym{Maj}_n$ can be implemented using $O(n)$ NAND gates.
\end{proposition}
\textit{By $G(n)$ we will denote the number of NAND gates required to implement $\sym{Maj}_n$.}
For some concrete values of $n$, estimates of $G(n)$ were obtained in~\cite[Section~V]{DBLP:journals/tit/CarletS26}. The strategy used was the following. 
A circuit to implement $\sym{Maj}_n$ has two parts, one part to compute the weight of an $n$-bit string, and the other part to determine whether
the weight is more than $\floor{n/2}$. The second part requires $O(\log_2 n)$ gates, and $G(n)$ is determined mainly by the size of the circuit to 
compute the weight of an $n$-bit string. 
Using the algorithm in~\cite{DBLP:journals/eccc/ECCC-TR05-049} for finding the weight of a binary string, it was shown in~\cite[Section~V]{DBLP:journals/tit/CarletS26} that the circuit to compute the weight of an $n$-bit string can be implemented 
using full and half adders. It turns out that the total of number of full and half adders is less than $n$; 
for example, it was shown that the numbers of full and half adders for $n=37$ are 31 and 3 respectively -- see~\cite[Section~V]{DBLP:journals/tit/CarletS26} for other examples.
A full adder requires 9 NAND gates, while a half adder requires 5 NAND gates. 
So overall $G(n)$ is at most about $10n$ NAND gates. 

\paragraph{Threshold function.} For $n\geq 1$, and $1\leq d\leq n$, let $\sym{Th}_{n,d}:\{0,1\}^n\rightarrow \{0,1\}$ be the threshold
function defined in the following manner. For $\mathbf{x}\in\{0,1\}^n$, $\sym{Th}_{n,d}(\mathbf{x})=0$ if and only if $\sym{wt}(\mathbf{x})<d$. 
Note that $\sym{Maj}_n=\sym{Th}_{n,\lceil n/2\rceil}$. The following result provides the algebraic immunity of $\sym{Th}_{n,d}$ 
(see also~\cite[Lemma~4]{DBLP:journals/dam/Meaux22}).
\begin{proposition}[Lemma~6 of~\cite{10.1007/978-3-030-36237-9-10}] \label{prop-th-ai}
Let $n>1$ and $1\leq d\leq n$ be positive integers. Then 
	$\sym{AN}(\sym{Th}_{n,d}) = n-d+1$, $\sym{AN}(1\oplus \sym{Th}_{n,d}) = d$. Hence, $\sym{AI}(\sym{Th}_{n,d})$, and
	$\Delta_{\sym{AN}}(\sym{Th}_{n,d})=\mid n+1-2d\mid$.
\end{proposition}


\paragraph{Direct sum.}
A simple way to construct a function is to add together two functions on disjoint sets of variables. The constructed function is called the direct sum of the
two smaller functions. Let $n_1$ and $n_2$ be positive integers and $g$ and $h$ be functions on $n_1$ and $n_2$ variables respectively. Define
\begin{eqnarray}\label{eqn-dsum}
	f(X_1,\ldots,X_{n_1},Y_1,\ldots,Y_{n_2}) & = & g(X_1,\ldots,X_{n_1}) \oplus h(Y_1,\ldots,Y_{n_2}).
\end{eqnarray}
Basic properties of a function constructed as a direct sum is given in the following result.
\begin{proposition}[Proposition~1 of~\cite{DBLP:conf/eurocrypt/SarkarM00}]\label{prop-dsum}
	Let $f$ be constructed as in~\eqref{eqn-dsum}. Then
	\begin{compactenum}
	\item $\sym{deg}(f)=\max\{ \sym{deg}(g),\sym{deg}(h) \}$.
	\item $\sym{nl}(f) = 2^{n_1}\sym{nl}(h) + 2^{n_2}\sym{nl}(g) - 2\sym{nl}(g)\sym{nl}(h)$. Equivalently, $\sym{LB}(f)=\sym{LB}(g)\cdot \sym{LB}(h)$.
	\item $f$ is balanced if and only if at least one of $g$ or $h$ is balanced.
	\item If $g$ is $m_1$-resilient and $h$ is $m_2$-resilient, then $f$ is $m$-resilient, where $m=\max\{m_1,m_2\}$.
	\end{compactenum}
\end{proposition}
A special case of direct sum is the following simple construction: $f(W,\mathbf{X})$ is defined to be $f(W,\mathbf{X})=W\oplus g(\mathbf{X})$. 
The following properties of such an $f$ can be proved directly, and also follows as a consequence of Proposition~\ref{prop-dsum}:
$\sym{nl}(f)=2\cdot\sym{nl}(g)$ and $\sym{LB}(f)=\sym{LB}(g)$; $f$ is balanced, and if $g$ is $m$-resilient, then $f$ is $(m+1)$-resilient. 

Upper and lower bounds on the algebraic immunity of a function constructed as a direct sum are given by the following result.
\begin{proposition}[Lemma~3 of~\cite{DBLP:conf/eurocrypt/MeauxJSC16}]\label{prop-AI-dsum}
        For $f$ constructed as in~\eqref{eqn-dsum}, $\max\{\sym{AI}(g),\sym{AI}(h)\} \leq \sym{AI}(f)\leq \sym{AI}(g) + \sym{AI}(h)$.
\end{proposition}
A sharper lower bound on the algebraic immunity of direct sum is known.
\begin{proposition}[Corollary~3 of~\cite{DBLP:journals/dam/Meaux22}]\label{prop-AI-dsum-sharper} 
	For $f$ constructed as in~\eqref{eqn-dsum}, suppose $\sym{AI}(g)\geq \sym{AI}(h)$, then
	$\sym{AI}(g) + \min\{\Delta_{\sym{AN}}(g),\sym{AI}(h)\} \leq \sym{AI}(f)$.
\end{proposition}
If $h$ is a non-constant function, and $\Delta_{\sym{AN}}(g)$ is positive, then the lower bound in Proposition~\ref{prop-AI-dsum-sharper}
is greater than the lower bound in Proposition~\ref{prop-AI-dsum}.

\paragraph{Tunable direct sum.}
	The direct sum of different functions has been considered in a number of previous works~\cite{DBLP:conf/eurocrypt/MeauxJSC16,DBLP:conf/indocrypt/MeauxCJS19,DBLP:journals/tit/CarletM22,DBLP:conf/africacrypt/GerardGM24}. 
	Direct sum of monomials was considered in~\cite{DBLP:conf/eurocrypt/MeauxJSC16,DBLP:conf/indocrypt/MeauxCJS19}, the direct sum
	of linear, quadratic, and triangular functions was considered in~\cite{DBLP:conf/eurocrypt/MeauxJSC16}, the direct sum of linear
	and threshold functions was considered in~\cite{DBLP:journals/tit/CarletM22}, and the direct sum of linear, quadratic, and threshold
	functions was considered in~\cite{DBLP:conf/africacrypt/GerardGM24} (and also in the longer version of~\cite{DBLP:conf/indocrypt/MeauxCJS19}).
	Based on~\cite{DBLP:conf/africacrypt/GerardGM24} we define the following construction.
\begin{construction}\label{cons-tune-ds}
	Let $m\geq 0$, $k,s\geq 1$, and $1\leq d\leq s$ be integers. For $\mathbf{W}=(W_1,\ldots,W_m)$,
	$\mathbf{X}=(X_1,\ldots,X_k)$, $\mathbf{Y}=(Y_1,\ldots,Y_k)$, and $\mathbf{Z}=(Z_1,\ldots,Z_s)$, define a Boolean function $f$
	on $m+2k+s$ variables as follows.
	\begin{eqnarray}\label{eqn-tune-ds}
		f(\mathbf{W},\mathbf{X},\mathbf{Y},\mathbf{Z}) 
		& = & L_m(\mathbf{W}) \oplus Q_{2k}(\mathbf{X},\mathbf{Y}) \oplus \sym{Th}_{s,d}(\mathbf{Z}).
	\end{eqnarray}
\end{construction}
We state the following result in the case where $s$ is odd and $\sym{Th}_{s,d}=\sym{Maj}_s$. 
Note that for odd $s$, $\sym{Maj}_s(\mathbf{Z})$ is balanced. 
From Theorem~\ref{thm-maj-prop}, and Propositions~\ref{prop-dsum} and~\ref{prop-AI-dsum}, we have the following result.
\begin{proposition} \label{prop-tdc}
	In Construction~\ref{cons-tune-ds}, let $k>1$ and $s$ be odd and
	suppose $\sym{Maj}_s$ is used to instantiate $\sym{Th}_{s,d}$ in Construction~\ref{cons-tune-ds}, and let $f$ be the function that is obtained.
	Then we have the following.
	\begin{enumerate}
		\item $f$ is $m$-resilient.
		\item $\sym{nl}(f) = 
			2^m\left( 2^{2k}\left(2^{s-1} - {s-1 \choose \lfloor s/2\rfloor}\right) 
			+ 2^s (2^{2k-1}-2^{k-1}) 
			- 2(2^{2k-1}-2^{k-1})\left(2^{s-1} - {s-1 \choose \lfloor s/2\rfloor}\right)\right)$, and hence
			$\sym{LB}(f)=2^{-k} {s-1 \choose (s-1)/2} \big/2^{s-1}$.
		\item $f$ has algebraic immunity to be at least $\lceil s/2\rceil$.
		\item $f$ can be implemented using $k$ AND, $k+m$ XOR, and $G(s)$ NAND gates.
	\end{enumerate}
\end{proposition}
Applying Stirling's approximation ${2j \choose j} \approx 4^{j}/\sqrt{\pi j}$, we obtain 
\begin{eqnarray}\label{eqn-approx-LB}
	\mid\log_2 \sym{LB}(f)\mid & \approx & k+\tfrac12\log_2\bigl(\pi(s-1)/2\bigr). 
\end{eqnarray}
From~\eqref{eqn-approx-LB} one may note that the linear bias of any function obtained from Construction~\ref{eqn-approx-LB} 
is essentially determined by $k$ and the effect of $s$ is only logarithmic.

\begin{remark}\label{rem-th-maj}
	If instead of $\sym{Maj}_s$, Construction~\ref{cons-tune-ds} is instantiated with a more general $\sym{Th}_{s,d}$, then the
	parameters guaranteed in Proposition~\ref{prop-tdc} will change. For $k>0$, using Propositions~\ref{prop-th-ai} and~\ref{prop-AI-dsum-sharper}, the
	algebraic immunity of $Q_{2k}\oplus \sym{Th}_{s,d}$ (and hence the algebraic immunity of $f$) is at least 
	$1+\lceil s/2\rceil$ if $d=\lfloor s/2\rfloor$ (where we need to use the fact that the algebraic immunity of $Q_{2k}$ is 2). 
	The resulting $Q_{2k}\oplus \sym{Th}_{s,d}$, on the other hand, is not balanced, 
	and so the resiliency of $f$ is $m-1$. Also, the expression for the linear bias will change. Such changes in parameters arising from using a more 
	general $\sym{Th}_{s,d}$ instead of $\sym{Maj}_s$ in Construction~\ref{cons-tune-ds} will not affect the general nature of 
	the comparison between Construction~\ref{cons-tune-ds} and the constructions that we propose later.
\end{remark}

\paragraph{Maiorana-McFarland with Majority.}
The following lower bound on the algebraic immunity of a special class of Maiorana-McFarland bent functions was obtained in~\cite{DBLP:journals/tit/CarletS26}.
\begin{theorem}[Theorem~2 of~\cite{DBLP:journals/tit/CarletS26}]\label{thm-MMMaj-AI}
	Let $k\geq 2$, $n=2k$, and $\psi:\{0,1\}^k\rightarrow \{0,1\}^k$ be an affine map, i.e. each of the coordinate functions of $\psi$ is an affine function of the input variables.
	Then 
	\begin{eqnarray}\label{eqn-pi-h-ai}
		\sym{AI}((\psi,h)\mbox{-}\sym{MM}_{2k}) & \geq  & \sym{AI}(h).
	\end{eqnarray}

	Consequently, $\sym{AI}((\psi,\sym{Maj})\mbox{-}\sym{MM}_{2k})\geq \sym{AI}(\sym{Maj}_k)=\lceil k/2\rceil$.
\end{theorem}
Theorem~\ref{thm-MMMaj-AI} holds for any affine map $\psi$. From efficiency considerations, we will consider $\psi$ to be a bit permutation, i.e.
there is a permutation $\rho$ of $\{1,\ldots,k\}$ such that for any $(x_1,\ldots,x_k)\in\mathbb{F}_2^n$, $\psi(x_1,\ldots,x_k)=(x_{\rho(1)},\ldots,x_{\rho(k)})$. 
Implementation of a bit permutation requires only the proper connection pattern, and does not require any gates. For the case where $\psi$ is a bit permutation,
it was shown in~\cite{DBLP:conf/latincrypt/Meaux25} that the algebraic immunity of $(\psi,\sym{Maj})\mbox{-}\sym{MM}_{2k}$ is at most $1+2^{\lceil\log(n/4)\rceil}$.

\paragraph{Gate count.} From an implementation point of view, it is of interest to obtain functions which are efficient to implement. A measure of
implementation efficiency is the circuit size which is given by the number of 2-input gates required to implement an $n$-variable function $f$. The truth table 
representation of $f$ requires a circuit size of $\Omega(2^n)$. For even moderate values of $n$, such a representation results in a very large circuit. 
In this paper, we consider the following 2-input gates: XOR, AND, and NAND. The size of the circuit to implement the majority function is considered
in terms of NAND gates. For the other components we use only XOR and AND gates. 

We note that the circuit size is only one measure of efficiency. This measure is relevant for usual stream cipher applications. On the other hand,
for applications to homomorphic encryption, a relevant measure is the multiplicative depth. The works~\cite{DBLP:conf/eurocrypt/MeauxJSC16,DBLP:conf/indocrypt/MeauxCJS19,DBLP:journals/tit/CarletM22,DBLP:conf/africacrypt/GerardGM24} are focused towards homomorphic encryption. 

\begin{remark}\label{rem-gen-th}
In the later sections, we describe constructions which are based on $(\psi,\sym{Maj})\mbox{-}\sym{MM}_{2k}$. All of these constructions can 
	also be based on $(\psi,\sym{Th}_{k,d})\mbox{-}\sym{MM}_{2k}$ for some suitable more homomorphic friendly choice of $\sym{Th}_{k,d}$. 
\end{remark}

\section{Basic Constructions of $m$-Resilient Functions \label{sec-basic}}
In the present section, we provide two basic constructions of $m$-resilient functions with guarantees on linear bias and algebraic immunities. 

\subsection{Case $n\not\equiv m\bmod 2$ \label{subsec-basic-n-m-not-same-parity}}
The following construction builds on the basic fact that adding new variables increases the order of resiliency. 
\begin{construction}\label{cons-basic-1}
	Let $m$ be a non-negative integer, and $n>m$ be another integer such that $n\not\equiv m\bmod 2$. Let $k=(n-m-1)/2$.
	Let $\psi:\{0,1\}^k\rightarrow \{0,1\}^k$ be a bit permutation. 
	The $n$-variable function $f$ is defined as follows.
	\begin{eqnarray}\label{eqn-res-even-gen}
		\lefteqn{f(X_1,\ldots,X_{m+1},U_1,\ldots,U_k,V_1,\ldots,V_k)} \nonumber \\
		& = & X_1\oplus \cdots\oplus X_{m+1}\oplus (\psi,h)\mbox{-}\sym{MM}_{2k}(U_1,\ldots,U_k,V_1,\ldots,V_k).
	\end{eqnarray}
\end{construction}
The properties of the function obtained from Construction~\ref{cons-basic-1} are stated in the following result. 
\begin{proposition}\label{prop-direct-m>1-even}
	Let $f$ be the $n$-variable function obtained as in Construction~\ref{cons-basic-1}.
	Then $f$ is an $m$-resilient function with linear bias equal to $2^{-(n-m-1)/2}$. Further, 
		if $h=\sym{Maj}_k$, then the algebraic immunity of $f$ is at least $\lceil(n-m-1)/4 \rceil$, and $f$ can be implemented using 
		$k$ AND, $k+m+1$ XOR, and $G(k)$ NAND gates.
\end{proposition}
\begin{proof}
	Since $m+1$ new variables are added to a bent function, the order of resiliency is $m$. The linear bias of the bent function on $2k$ variables
	is $2^{-k}$, where $k=(n-m-1)/2$, and since the new variables are simply added, the linear bias remains unchanged. From Proposition~\ref{prop-AI-dsum}
	the algebraic immunity of $f$ is at least the algebraic immunity of the bent function; from Theorem~\ref{thm-MMMaj-AI} the algebraic immunity 
	of the bent function is at least the algebraic immunity of $h=\sym{Maj}_k$; and from Theorem~\ref{thm-maj-prop} the algebraic immunity
	of $\sym{Maj}_k$ is equal to $\lceil k/2 \rceil$. 
	$\sym{Maj}_k$ requires $G(k)$ NAND gates. The counts of the other gates follow easily from Construction~\ref{cons-basic-1}.
\end{proof}

\paragraph{Comparison to Proposition~\ref{prop-tdc}.}
Both Propositions~\ref{prop-tdc} and~\ref{prop-direct-m>1-even} increase resiliency by adding variables. 
Further, both Propositions~\ref{prop-tdc} and~\ref{prop-direct-m>1-even} rely on the majority 
function to achieve a guaranteed lower bound on algebraic immunity. So for the purpose of comparison, we assume that the order of resiliency, and the
lower bound on the algebraic immunity is the same for both the cases. 

Let $a_0$ be the target lower bound on the algebraic immunity. Then in Proposition~\ref{prop-tdc} we have $s=2a_0-1$, and in 
Proposition~\ref{prop-direct-m>1-even} we have $k=2a_0-1$. In Table~\ref{tab:prop5} we compare the parameters of 
Propositions~\ref{prop-tdc} and~\ref{prop-direct-m>1-even}. The column labelled `difference' provides the difference between the parameters in 
Proposition~\ref{prop-tdc} and~\ref{prop-direct-m>1-even}.

Fixing $a_0$ and $m$ fixes all the parameters in Proposition~\ref{prop-direct-m>1-even}, while in Proposition~\ref{prop-tdc}, the parameter $k$ is still free. 
As $k$ increases, for Proposition~\ref{prop-tdc} both the number of variables and the circuit size increase, while the linear bias decreases 
(equivalently, $\lvert\log_2\sym{LB}\rvert$ increases). The comparison between Propositions~\ref{prop-tdc} and~\ref{prop-direct-m>1-even} is 
based on $k$ and is summarised as follows.
\begin{enumerate}
	\item For $k<a_0$, Proposition~\ref{prop-tdc} has both smaller number of variables, and smaller circuit size, but the linear bias is much 
		larger; in Proposition~\ref{prop-tdc} it is $2^{-k-c}$, while in Proposition~\ref{prop-direct-m>1-even} it is $2^{-2a_0+1}$.
	\item For $k=a_0$, both Propositions~\ref{prop-tdc} and~\ref{prop-direct-m>1-even} have the same number of variables; 
		Proposition~\ref{prop-tdc} has a circuit which is smaller by $2a_0-1$ AND and XOR gates; and the linear bias of
		Proposition~\ref{prop-tdc} is $2^{-(a_0+c)}$, and that of Proposition~\ref{prop-direct-m>1-even} is $2^{-2a_0+1}$. To understand
		this concretely, suppose $a_0=30$, then Proposition~\ref{prop-direct-m>1-even} requires 59 more AND and XOR gates than 
		Proposition~\ref{prop-tdc}, but Proposition~\ref{prop-direct-m>1-even} has linear bias $2^{-59}$, whereas Proposition~\ref{prop-tdc} has 
		linear bias $2^{-30}/\sqrt{29\pi}$. So compared to Proposition~\ref{prop-tdc}, a few extra gates required in 
		Proposition~\ref{prop-direct-m>1-even} lead to a substantial reduction in the linear bias.
	\item For $k=\lceil 2a_0-1-c\rceil$, both Propositions~\ref{prop-tdc} and~\ref{prop-direct-m>1-even} have almost the same linear bias
		and also almost the same circuit size. The number of variables required by Proposition~\ref{prop-tdc} is about $2a_0-2-2c$ more than
		that of Proposition~\ref{prop-direct-m>1-even}. Again taking $a_0=30$, we see that Proposition~\ref{prop-tdc} requires about 50 more
		variables compared to Proposition~\ref{prop-direct-m>1-even}.
	\item As $k$ increases beyond $\lceil 2a_0-1-c\rceil$, in comparison to Proposition~\ref{prop-direct-m>1-even}, the linear bias of 
		Proposition~\ref{prop-tdc} goes down further, at the cost of increase in the number of variables and the circuit size. 
\end{enumerate}

\begin{table}[h]
\centering
\small\setlength{\tabcolsep}{6pt}\renewcommand{\arraystretch}{1.25}
\begin{tabular}{@{}l c c c@{}}
\toprule
	& Proposition~\ref{prop-tdc} & Proposition~\ref{prop-direct-m>1-even} & difference \\
 & $s=2a_0-1$, $k$ free & $k=2a_0-1$ & \\
\midrule
$n$ & $m+2k+2a_0-1$ & $m+4a_0-1$ & $2k-2a_0$\\
resiliency & $m$ & $m$ & $0$\\
	$\lvert\log_2\sym{LB}\rvert$ & $\approx k+c$ & $2a_0-1$ & $\approx k+c-2a_0+1$\\
	$\sym{AI}$ (lower bound) & $a_0$ & $a_0$ & $0$\\
\midrule
\#AND & $k$ & $2a_0-1$ & $k-2a_0+1$\\
\#XOR & $k+m$ & $2a_0+m$ & $k-2a_0$\\
\midrule
size $-\,G(2a_0-1)$ & $2k+m$ & $4a_0+m-1$ & $2k-4a_0+1$\\
\bottomrule
\end{tabular}
	\caption{\label{tab:prop5}Parameters of Proposition~\ref{prop-tdc} and Proposition~\ref{prop-direct-m>1-even}.
	Here $c=\tfrac12\log_2(\pi(a_0-1))$ (see~\eqref{eqn-approx-LB}), and $G(2a_0-1)$ is the number of NAND gates required for $\mathrm{Maj}_{2a_0-1}$.}
\end{table}

To summarise, in Proposition~\ref{prop-tdc} the parameter $k$ provides a separate control. From Point~4 above, for the same order of resiliency,
the same lower bound on algebraic immunity, similar linear bias and efficiency, Proposition~\ref{prop-tdc} requires substantially more variables
compared to Proposition~\ref{prop-direct-m>1-even}.

\paragraph{Balanced functions.}

The special case of balanced functions was considered in~\cite{DBLP:journals/tit/CarletS26} and is given in the following result.
\begin{corollary}[\cite{DBLP:journals/tit/CarletS26}]\label{cor-bal} 
	Let $n\equiv 1\bmod 2$. It is possible to construct an $n$-variable, balanced function with linear bias equal to $2^{-(n-1)/2}$ (which is almost optimal), 
	algebraic immunity at least $\lceil(n-1)/4 \rceil$, and can be implemented using $(n-1)/$ AND, $(n+1)/2$ XOR, and $G((n-1)/2)$ NAND gates.
\end{corollary}
\begin{proof}
Putting $m=0$ in Proposition~\ref{prop-direct-m>1-even} provides the result.
\end{proof}

Several papers~\cite{DBLP:conf/asiacrypt/CarletF08,DBLP:journals/tit/WangPKX10,DBLP:journals/tit/Carlet11,DBLP:journals/tit/TangCT13,DBLP:journals/jossac/ShanHZL18,DBLP:journals/dcc/WangS19,DBLP:journals/dam/Hu0H20} proposed constructions of $n$-variable balanced functions which achieve optimal algebraic immunity $\lceil n/2\rceil$
with provable upper bounds on the linear bias. The linear bias of the $n$-variable function constructed in the last of this line of papers, i.e. in~\cite{DBLP:journals/dam/Hu0H20},
is at most $(0.402963+\ln 2/\pi)2^{-(n-2)/2}+\delta$, where $\delta$ is a small (though quite complicated) quantity.

For odd $n$, Corollary~\ref{cor-bal} provides constructions of $n$-variable balanced functions with almost
optimal linear bias of $2^{-(n-1)/2}$ and a lower bound of $\lceil (n-1)/4\rceil$ on algebraic immunity. The linear biases of the functions obtained from the previous constructions
(both the upper bound as well as actual values for concrete functions) are higher than the linear biases of the functions obtained from
Corollary~\ref{cor-bal}. So Corollary~\ref{cor-bal} and the previous constructions provide different points on the nonlinearity/linear bias trade-off curve
for balanced functions. 

The main advantage of the functions constructed using Corollary~\ref{cor-bal} is that these functions can be 
constructed using $O(n)$ gates, while the functions constructed using the methods in~\cite{DBLP:conf/asiacrypt/CarletF08,DBLP:journals/tit/WangPKX10,DBLP:journals/tit/Carlet11,DBLP:journals/tit/TangCT13,DBLP:journals/jossac/ShanHZL18,DBLP:journals/dcc/WangS19,DBLP:journals/dam/Hu0H20} are essentially based on discrete 
logarithm computation and require super-polynomial size circuits. To see this advantage in concrete terms, we consider the case of the 28-variable function 
(say $f_{28}$) reported in Table~2 of~\cite{DBLP:journals/dam/Hu0H20} which has algebraic immunity
14 and nonlinearity 134201460 (equivalently linear bias equal to about $2^{-13.01}$). According to the description of hardware implementation in 
Section~4.2 of~\cite{DBLP:journals/dam/Hu0H20}, implementation of $f_{28}$ will require a look-up table of size $2^{28}$ along with other gates. 
Taking $n=55$ in Corollary~\ref{cor-bal}, we obtain a 55-variable balanced function (say $f_{55}$) having linear bias 
$2^{-27}$ and algebraic immunity at least 14. Further, $f_{55}$ can be implemented using 27 AND, 28 XOR, and $G(27)$ NAND gates. 
In comparison to $f_{28}$, the function $f_{55}$ has a much lower linear bias and is much more efficient to implement. 



\paragraph{1-Resilient functions.}
The special case of 1-resilient function obtained from Proposition~\ref{prop-direct-m>1-even} is given in the following result.
\begin{corollary}\label{cor-1-res-a}
	Let $n\equiv 0\bmod 2$. It is possible to construct an $n$-variable, 1-resilient function with linear bias equal to $2^{-(n-2)/2}$ (which is almost optimal), 
	algebraic immunity at least $\lceil(n-2)/4 \rceil$, and can be implemented using $(n-2)/2$ AND, $n/2$ XOR, and $G((n-2)/2)$ NAND gates.
\end{corollary}
\begin{proof}
Putting $m=1$ in Proposition~\ref{prop-direct-m>1-even} provides the result.
\end{proof}

For even $n$, several papers~\cite{DBLP:journals/dam/TuD12,DBLP:conf/cisc/WangLL12,DBLP:journals/ijfcs/TangCT14,DBLP:journals/ijcm/WangZWZW16,DBLP:journals/tit/TangCTZ17} 
have proposed constructions of 1-resilient functions with optimal algebraic immunity $n/2$ and provable upper bounds on linear bias. To the best of our knowledge the 
last of such results is~\cite{DBLP:journals/tit/TangCTZ17} which improves upon works prior to it by providing lower linear bias, and also guarantees the fast algebraic immunity
to be at least $n-6$. The upper bound on the linear bias of $n$-variable 
functions constructed in~\cite{DBLP:journals/tit/TangCTZ17} is $((1+(n/2)\ln 2)/\pi + (\pi+16)/32)2^{-n/2 + 1}+2^{-(n-1)}$. 
The linear bias of the functions constructed in Corollary~\ref{cor-1-res-a} is $2^{-(n-2)/2}$ which is almost optimal, the algebraic immunity is at 
least $\lceil(n-2)/4 \rceil$, and hence the fast algebraic immunity is at least $1+\lceil(n-2)/4 \rceil$. 
So for 1-resilient functions, the functions constructed in~\cite{DBLP:journals/tit/TangCTZ17} and those in 
Corollary~\ref{cor-1-res-a} represent two distinct trade-off points with respect to algebraic immunity, fast algebraic immunity and linear bias. 

The main advantage of the functions constructed using Corollary~\ref{cor-1-res-a} is that these functions can be
constructed using $O(n)$ gates, while the method in~\cite{DBLP:journals/tit/TangCTZ17} is based on defining the support of the desired Boolean function using 
powers of a primitive element over the field $\mathbb{F}_{2^{n/2}}$; in particular, from the description of Construction~2 
in~\cite{DBLP:journals/tit/TangCTZ17} it appears that the only reasonable method of implementing the function is to use a truth table requiring $\Omega(2^n)$ 
gates.

We provide a concrete example to highlight the difference between the functions constructed using Corollary~\ref{cor-1-res-a} and 
Theorem~6 of~\cite{DBLP:journals/tit/TangCTZ17} (which is based on Construction~2 of~\cite{DBLP:journals/tit/TangCTZ17}).
Suppose $n=32$. Theorem~6 of~\cite{DBLP:journals/tit/TangCTZ17} provides a 32-variable, 1-resilient function (say $g_{32}$)
with algebraic immunity 16, fast algebraic immunity at least 26 and nonlinearity at least 2147192232 (equivalently,
linear bias at most about $2^{-12.85}$) which can be implemented using around $2^{32}$ gates. 
Taking $n=100$ in Corollary~\ref{cor-1-res-a} we obtain a 100-variable, 1-resilient function (say $g_{100}$) with algebraic immunity at 
least 25 (and hence fast algebraic immunity at least 26), and almost optimal linear bias equal to $2^{-49}$. Further, $g_{100}$ can be implemented 
using 49 AND, 50 XOR, and $G(49)$ NAND gates.
So comparing $g_{32}$ to $g_{100}$, we see that $g_{100}$ has a much lower linear bias and is much more efficient to implement. The trade-off is that
the number of variables required is more.

\begin{remark}\label{rem-lo-bnd-not-strict}
	Proposition~\ref{prop-direct-m>1-even} guarantees algebraic immunity of the constructed $n$-variable, $m$-resilient function to be \textit{at least} $\lceil (n-m-1)/4\rceil$. 
	For concrete values of $n$, the actual value of algebraic immunity can actually be greater than the lower bound. For example, if we take $n=10$ and $m=1$,
	then the lower bound on algebraic immunity of the 10-variable function constructed using Proposition~\ref{prop-direct-m>1-even} is 2; we constructed the 10-variable
	function and computed its actual algebraic immunity which turns out to be 3.
\end{remark}

\subsection{Case $n\equiv m\bmod 2$ \label{subsec-basic-n-m-same-parity}}
Proposition~\ref{prop-direct-m>1-even} covers the case where $n$ and $m$ do not have the same parity. The case where $n$ and $m$ have the same parity 
can also be covered in a similar manner. For such a construction we use a 5-variable function given by the following result.

\begin{proposition}\label{prop-case-n=5}
Define
\begin{eqnarray}\label{eqn-n=5-simple}
	f_5(X_1,X_2,Z_1,Z_2,Z_3) & = & Z_1 \oplus Z_2 \oplus X_1(Z_1\oplus Z_3) \oplus X_2(Z_2\oplus Z_3) \oplus X_1X_2(Z_1\oplus Z_2\oplus Z_3).
\end{eqnarray}
	The function $f_5$ is a 5-variable, 1-resilient function having degree 3, algebraic immunity 2, nonlinearity 12 (and hence almost optimal linear bias equal to $2^{-2}$),
	and can be implemented using 7 XOR and 4 AND gates. 
\end{proposition}
\begin{proof}
	We note that $f_5$ can be written in the following manner.
\begin{eqnarray*}\label{eqn-n=5}
f_5(X_1,X_2,Z_1,Z_2,Z_3) & = & (1\oplus X_1)(1\oplus X_2)\big(Z_1\oplus Z_2\big) \oplus (1\oplus X_1)X_2\big(Z_1\oplus Z_3\big) \nonumber \\
&   & \oplus X_1(1\oplus X_2)\big(Z_2\oplus Z_3\big) \oplus X_1X_2\big(Z_1\oplus Z_2\oplus Z_3\big).
\end{eqnarray*}
This shows that $f_5$ is the concatenation of 4 distinct linear functions each of which is non-degenerate on at least 2 variables. It follows that
$f_5$ is 1-resilient. Let the 4 distinct linear functions making up $f$ be $h_1,h_2,h_3$ and $h_4$. Any 5-variable affine function $\ell$ can be written 
as concatenation of $g_1,g_2,g_3,g_4$, where each $g_i$ is either $g$ or $1\oplus g$, where $g$ is some 3-variable linear function. 
	So $d(f,\ell)=\sum_{i=1}^4 d(h_i,g_i)$. If neither $g$ nor $1\oplus g$ is equal to any of the $h_i$'s, then $d(f,\ell)=16$;
	if $g$ or $1\oplus g$ is equal to one of the $h_i$'s, then $d(f,\ell)=12$; and if $g$ or $1\oplus g$ is equal to one of the $(1\oplus h_i)$'s,
	then $d(f,\ell)=20$. Since these are the only three possibilities, it follows that the nonlinearity of $f$ is 12 and hence $f_5$
	has almost optimal linear bias equal to $2^{-2}$. 

	The degree of $f_5$ is 3. We computed the algebraic immunity of $f_5$ using a computer and found it to be 2.

The expression for $f_5$ given by~\eqref{eqn-n=5-simple} can be implemented using 7 XOR and 4 AND gates.
\end{proof}


Using Proposition~\ref{prop-case-n=5}, we provide the construction for the case where $n$ and $m$ have the same parity.
\begin{construction} \label{cons-basic-2}
	Let $m$ be a positive integer, and let $n\geq m+4$ be such that $n\equiv m\bmod 2$. Let $k=(n-m-4)/2$.
	Let $\psi:\{0,1\}^k\rightarrow \{0,1\}^k$ be a bit permutation. Let $f_5$ be the function defined in~\eqref{eqn-n=5-simple}. 
	The $n$-variable function $f$ is defined as follows.
	\begin{eqnarray}
		\lefteqn{f(Y_1,\ldots,Y_{m-1},X_1,X_2,Z_1,Z_2,Z_3,U_1,\ldots,U_k,V_1,\ldots,V_k)} \nonumber \\
		& = & Y_1\oplus \cdots\oplus Y_{m-1} \oplus f_5(X_1,X_2,Z_1,Z_2,Z_3)\oplus (\psi,h)\mbox{-}\sym{MM}_{2k}(U_1,\ldots,U_k,V_1,\ldots,V_k).
	\end{eqnarray}
\end{construction}
The properties of the function obtained from Construction~\ref{cons-basic-2} are stated in the following result.
\begin{proposition}\label{prop-direct-m>1-odd}
	Let $f$ be the $n$-variable function obtained as in Construction~\ref{cons-basic-1}.
	Then $f$ is an $m$-resilient function with linear bias equal to $2^{-(n-m)/2}$. Further, 
		if $h=\sym{Maj}_k$, then the algebraic immunity of $f$ is at least $\lceil(n-m-4)/4 \rceil$, and $f$ can be implemented using 
		$(k+4)$ AND, $(k+m+7)$ XOR, and $G(k)$ NAND gates.
\end{proposition}
\begin{proof}
	The function $f_5$ and $m-1$ new variables are added to the bent function. The function $f_5$ is itself 1-resilient, and adding the $m-1$ new variables increases 
	resiliency to $m$. The linear bias of the bent function on $2k$ variables is $2^{-k}$, where $k=(n-m-4)/2$. The linear bias of $f_5$ is $2^{-2}$, and so the
	linear bias of the sum of $f_5$ and the bent function is $2^{-k-2}$. Addition of the new variables does not change the linear bias. 
	
	From Proposition~\ref{prop-AI-dsum} the algebraic immunity of $f$ is at least the algebraic immunity of the bent function; from Theorem~\ref{thm-MMMaj-AI} the 
	algebraic immunity of the bent function is at least the algebraic immunity of $h=\sym{Maj}_k$; and from Theorem~\ref{thm-maj-prop} the algebraic immunity
	of $\sym{Maj}_k$ is equal to $\lceil k/2 \rceil$. 

	$\sym{Maj}_k$ requires $G(k)$ NAND gates. The counts of the other gates follow from directly from Construction~\ref{cons-basic-2} 
	and Proposition~\ref{prop-case-n=5}.
\end{proof}



\section{Iterated Construction \label{sec-iterate}}
Both Propositions~\ref{prop-direct-m>1-even} and~\ref{prop-direct-m>1-odd} essentially simply add new variables to increase resiliency. 
This creates linear structures (see Proposition~28 of~\cite{BF-book}) and may be considered undesirable.
In this section, we describe a different method for increasing resiliency. This results in functions which provide a different type of trade-off
between resiliency, nonlinearity, and algebraic immunity. 
For the construction, we resurrect the idea of an iterated construction of resilient functions which was only 
briefly sketched in~\cite{PASALIC2001158}. The idea in~\cite{PASALIC2001158} was itself based on a more general theoretical result from~\cite{DBLP:conf/indocrypt/Tarannikov00}. 
The description in~\cite{PASALIC2001158} briefly considered resiliency and nonlinearity, but
not algebraic immunity (in fact, the work~\cite{PASALIC2001158} predates the introduction of the notion of algebraic immunity). 

While the construction idea is from~\cite{PASALIC2001158}, the detailed description that we present below was not provided in~\cite{PASALIC2001158}.

\begin{method}\label{cons-iter}
	Let $g$ and $h$ be two $n$-variable functions, and let $f$ be an $(n+1)$-variable function obtained as
	$\sym{Concat}(g,h)$, i.e. $$f(X_1,\ldots,X_{n+1})=(1\oplus X_{n+1})g(X_1,\ldots,X_n) \oplus X_{n+1}h(X_1,\ldots,X_n).$$
	Define $(n+3)$-variable functions $G$ and $H$ as follows.
\begin{eqnarray}
	\lefteqn{G(X_1,\ldots,X_{n+3})} \nonumber \\
	& = & X_{n+3} \oplus X_{n+2} \oplus f(X_1,\ldots,X_{n+1}) \nonumber \\
	& = & X_{n+3} \oplus X_{n+2} \oplus (1\oplus X_{n+1})g(X_1,\ldots,X_n) \oplus X_{n+1}h(X_1,\ldots,X_n), \label{eqn-iterated-G} \\
	\lefteqn{H(X_1,\ldots,X_{n+3})} \nonumber \\
	& = & X_{n+3}\oplus X_{n+1} \oplus (1\oplus X_{n+3} \oplus X_{n+2}) g(X_1,\ldots,X_n) \oplus (X_{n+3}\oplus X_{n+2})h(X_1,\ldots,X_n). \label{eqn-iterated-H}
\end{eqnarray}
	By $\sym{Step}(g,h)$ we denote the pair of functions $(G,H)$ obtained in~\eqref{eqn-iterated-G} and~\eqref{eqn-iterated-H}.
	Define an $(n+4)$-variable function $F$ as $\sym{Concat}(G,H)$, i.e.
	\begin{eqnarray}\label{eqn-F}	
		F(X_1,\ldots,X_{n+4}) & = & (1\oplus X_{n+4})G(X_1,\ldots,X_{n+3}) \oplus X_{n+4}H(X_1,\ldots,X_{n+3}).
	\end{eqnarray}
\end{method}

The following result relates the properties of $g$ and $h$ to that of $G$, $H$ and $F$.

\begin{theorem}\label{thm-iterated-prop}
	Let $n$ be a positive integer, and $g$ and $h$ be two $n$-variable functions. Let $(G,H)=\sym{Step}(g,h)$ and $F=\sym{Concat}(G,H)$ be constructed as in 
	Method~\ref{cons-iter}. Then the following holds.
	\begin{compactenum}
		\item If $g$ and $h$ are $m$-resilient, then $G$ and $H$ are $(m+2)$-resilient.
		\item $\sym{nl}(G)=\sym{nl}(H)=4\sym{nl}(f)$, and equivalently $\sym{LB}(G)=\sym{LB}(H)=\sym{LB}(f)$.
		\item Let $\ell(X_1,\ldots,X_{n+3})$ be a linear function. Then either $G\oplus \ell$ or $H\oplus \ell$ is balanced.
		\item Obtaining $G$ and $H$ from $g$ and $h$ requires 8 XOR and 4 AND gates. 
		\item If $g$ and $h$ are $m$-resilient, then $F$ is $(m+2)$-resilient.
		\item $\sym{nl}(F)=2^{n+2}+\sym{nl}(G)=2^{n+2}+4\sym{nl}(H)$, and equivalently $\sym{LB}(F)=\sym{LB}(G)/2=\sym{LB}(H)/2$.
		\item Obtaining $F$ from $G$ and $H$ requires 2 XOR and 2 AND gates.
	\end{compactenum}
\end{theorem}
\begin{proof}
	\paragraph{Proof of the first two points.}
	If $g$ and $h$ are $m$-resilient, then so is $f$. Note that $G$ is obtained by adding two new variables to $f$. It then follows that $G$ is $(m+2)$-resilient
	and further $\sym{nl}(G)=4\sym{nl}(f)$. The function $H$ is obtained from the function $G$ by the following invertible linear transformation on the
	variables: $X_{n+1}\rightarrow X_{n+3}\oplus X_{n+2}$, $X_{n+2}\rightarrow X_{n+1}$, $X_{n+3}\rightarrow X_{n+3}$. 
	Since nonlinearity is invariant under an invertible linear transformation on the variables, it follows that $\sym{nl}(H)=\sym{nl}(G)$.

	The function $H(X_1,\ldots,X_{n+3})$ can be written as 
	$H(X_1,\ldots,X_{n+3})=X_{n+1}\oplus H^\prime(X_1,\allowbreak \ldots,\allowbreak X_{n},\allowbreak X_{n+2},\allowbreak X_{n+3})$, where
\begin{eqnarray*}
	\lefteqn{H^\prime(X_1,\ldots,X_{n},X_{n+2},X_{n+3})} \nonumber \\
	& = & (1\oplus X_{n+3})(1\oplus X_{n+2})g(X_1,\ldots,X_n) \oplus (1\oplus X_{n+3})X_{n+2}h(X_1,\ldots,X_n) \nonumber \\
	&  & \oplus X_{n+3}(1\oplus X_{n+2})(1\oplus h(X_1,\ldots,X_n)) \oplus X_{n+3}X_{n+2}(1\oplus g(X_1,\ldots,X_n)).
\end{eqnarray*}
	Since $H$ is obtained from $H^\prime$ by adding a new variable, it follows that
	to show $H$ is $(m+2)$-resilient, it is sufficient to show $H^\prime$ is $(m+1)$-resilient. To show that $H^\prime$ is $(m+1)$-resilient,
	it is sufficient to show that $H^\prime\oplus \ell^\prime$ is balanced for all $(n+2)$-variable linear functions 
	$\ell^\prime(X_1,\ldots,X_n,X_{n+2},X_{n+3})$ which are non-degenerate on at most $m+1$ variables. 
	Write
\begin{eqnarray*}
	\lefteqn{\ell^\prime(X_1,\ldots,X_n,X_{n+2},X_{n+3})} \nonumber \\
	& = & (1\oplus X_{n+3})(1\oplus X_{n+2})\ell_1(X_1,\ldots,X_n) \oplus (1\oplus X_{n+3})X_{n+2}\ell_2(X_1,\ldots,X_n) \nonumber \\
	& & \oplus X_{n+3}(1\oplus X_{n+2})\ell_3(X_1,\ldots,X_n) \oplus X_{n+3}X_{n+2} \ell_4(X_1,\ldots,X_n),
\end{eqnarray*}
where $\ell_1,\ell_2,\ell_3,\ell_4$ are linear functions. Note that for $1\leq i<j\leq 4$, either $\ell_i=\ell_j$ or $\ell_i=1\oplus \ell_j$. We have 
	\begin{eqnarray} \label{eqn-wt}
		\sym{wt}(H^\prime \oplus \ell^\prime) 
		& = & \sym{wt}(g\oplus \ell_1) + \sym{wt}(h\oplus \ell_2) + (2^n-\sym{wt}(h\oplus \ell_3)) + (2^n-\sym{wt}(g\oplus \ell_4)).
	\end{eqnarray}
	Suppose that $\ell^\prime$ is degenerate on both $X_{n+2}$ and $X_{n+3}$. In this case, all the $\ell_i$'s are equal, and so from~\eqref{eqn-wt}
	it follows that $\sym{wt}(H^\prime \oplus \ell)=2^{n+1}$. Next suppose that $\ell$ is non-degenerate on at least one of $X_{n+2}$ or $X_{n+3}$. 
	In this case, each of the $\ell_i$'s is non-degenerate
	on at most $m$ variables. Since both $g$ and $h$ are $m$-resilient, it follows that $\sym{wt}(g\oplus \ell_1)$, $\sym{wt}(h\oplus \ell_2)$,
	$\sym{wt}(h\oplus \ell_3)$ and $\sym{wt}(g\oplus \ell_4)$ are all equal to $2^{n-1}$, and so $\sym{wt}(H^\prime \oplus \ell)=2^{n+1}$.
	This shows that $H^\prime$ is $(m+1)$-resilient and hence $H$ is $(m+2)$-resilient.

	\paragraph{Proof of the third point.}
	Suppose the linear function $\ell(X_1,\ldots,X_{n+3})$ is degenerate on $X_{n+3}$.
	Then $G\oplus \ell$ can be written as $X_{n+3}$ plus a function which does not involve $X_{n+3}$, and hence $G\oplus \ell$ is balanced.
	Similarly, if $\ell$ is degenerate on $X_{n+2}$, then $G\oplus \ell$ can be written as $X_{n+2}$ plus a function which does not
	involve $X_{n+2}$, and hence $G\oplus \ell$ is balanced. Further, if $\ell$ is degenerate on $X_{n+1}$, then $H\oplus \ell$ can
	be written as $X_{n+1}$ plus a function which does not involve $X_{n+1}$, and hence $H\oplus \ell$ is balanced. So suppose
	$\ell$ is non-degenerate on all three of the variables $X_{n+1},X_{n+2}$ and $X_{n+3}$, i.e. 
	$\ell(X_1,\ldots,X_{n+3})=X_{n+3}\oplus X_{n+2}\oplus X_{n+1}\oplus \sigma(X_1,\ldots,X_n)$, where $\sigma$ is a linear function.  
	We argue that $H\oplus \ell$ is balanced. 
	From the definition of $H$, we have $H\oplus \ell=X_{n+2}\oplus (1\oplus X_{n+3} \oplus X_{n+2})g \oplus (X_{n+3}\oplus X_{n+2})h \oplus \sigma$
	(which is degenerate on $X_{n+1}$). Considering the four possible values of $X_{n+3}$ and $X_{n+2}$, the sub-functions of $H\oplus \ell$
	that are obtained are $g\oplus \sigma$ (corresponding to $X_{n+3}=0$, $X_{n+2}=0$), $1\oplus h\oplus \sigma$ (corresponding to
	$X_{n+3}=0$, $X_{n+2}=1$), $h\oplus \sigma$ (corresponding to $X_{n+3}=1$, $X_{n+2}=0$), and $1\oplus g\oplus \sigma$ (corresponding to $X_{n+3}=1$, $X_{n+2}=1$).
	So $\sym{wt}(H\oplus \ell)=2(\sym{wt}(g\oplus \sigma)+\sym{wt}(1\oplus h\oplus \sigma)+\sym{wt}( h\oplus \sigma)+\sym{wt}(1\oplus g\oplus \sigma))=2^{n+2},$
	where the factor of $2$ arises from the variable $X_{n+1}$ on which $H\oplus \ell$ is degenerate. Hence, $H\oplus \ell$ is balanced.

	\paragraph{Proof of the fourth point.}
	The counts of the XOR and AND gates are clear from~\eqref{eqn-iterated-G} and~\eqref{eqn-iterated-H}.

	\paragraph{Proof of the fifth and sixth points.}
	Next we provide the arguments for the properties of $F$. 
	Since $H$ and $G$ are $(m+2)$-resilient, it follows that $F$ is also $(m+2)$-resilient. Let $\mu(X_1,\ldots,X_{n+4})$ be an affine
	function on $(n+4)$ variables. Then $\mu$ can be written as 
	$\mu=c\oplus dX_{n+4}\oplus \ell(X_1,\ldots,X_{n+3})$, where $c,d\in \mathbb{F}_2$, and $\ell$ is a linear function on $n+3$ variables.
	Then $\sym{wt}(F\oplus \mu)=\sym{wt}(c\oplus G\oplus \ell) + \sym{wt}(c\oplus d \oplus H\oplus \ell)$. 
	From the third point of the theorem, we have that either $G\oplus \ell$, or $H\oplus \ell$ is balanced. Suppose $G\oplus \ell$ is balanced.
	Then $\sym{wt}(F\oplus \mu)=2^{n+2}+ \sym{wt}(c\oplus d \oplus H\oplus \ell)\geq 2^{n+2}+\sym{nl}(H)$, where the equality is achieved for $c,d$ and 
	$\ell$ such that $\sym{nl}(H)=\sym{wt}(c\oplus d \oplus H\oplus \ell)$. Similarly, if $H\oplus \ell$ is balanced, then
	$\sym{wt}(F\oplus \mu)=2^{n+2}+ \sym{wt}(c\oplus G\oplus \ell)$, where the equality is achieved for $c$ and $\ell$ such that
	$\sym{nl}(G)=\sym{wt}(c \oplus G\oplus \ell)$. From these two cases, the statement on the nonlinearity of $F$ follows.

	\paragraph{Proof of the seventh point.} This is clear from the construction.
\end{proof}

To obtain a lower bound on the algebraic immunities of $G$ and $H$ obtained from Method~\ref{cons-iter}, we first prove the following general result.

\begin{proposition}\label{prop-sub-function}
Let $n$, $n_1$ and $n_2$ be positive integers with $n=n_1+n_2$, and let $f(X_1,\ldots,X_{n_1},Y_1,\ldots,Y_{n_2})$ be an $n$-variable function. We write
\begin{eqnarray} \label{eqn-decompose}
	f(X_1,\ldots,X_{n_1},Y_1,\ldots,Y_{n_2})
        & = & \bigoplus_{\bm{\alpha}=(\alpha_1,\ldots,\alpha_{n_2})} (1\oplus \alpha_1\oplus Y_1)\cdots (1\oplus \alpha_{n_2}\oplus Y_{n_2}) f_{\bm{\alpha}}(X_1,\ldots,X_{n_1}),
\end{eqnarray}
where for $\bm{\alpha}\in\mathbb{F}_2^{n_2}$, 
	$f_{\bm{\alpha}}(X_1,\ldots,X_{n_1})=f(X_1,\ldots,X_{n_1},\alpha_1,\ldots,\alpha_{n_2}).$
Then $\sym{AI}(f) \geq  \displaystyle \min_{\bm{\alpha}\in\mathbb{F}_2^{n_2}} \sym{AI}(f_{\bm{\alpha}}).$
\end{proposition}
\begin{proof}
Suppose $g$ is a non-zero annihilator of $f$. We write
	\begin{eqnarray*}
       g(X_1,\ldots,X_{n_1},Y_1,\ldots,Y_{n_2})
        & = & \bigoplus_{\bm{\alpha}=(\alpha_1,\ldots,\alpha_{n_2})} (1\oplus \alpha_1\oplus Y_1)\cdots (1\oplus \alpha_{n_2}\oplus Y_{n_2}) g_{\bm{\alpha}}(X_1,\ldots,X_{n_1}),
	\end{eqnarray*}
	where $g_{\bm{\alpha}}(X_1,\ldots,X_{n_1})=g(X_1,\ldots,X_{n_1},\alpha_1,\ldots,\alpha_{n_2})$.

	Since $gf=0$, it follows that $g_{\bm{\alpha}}f_{\bm{\alpha}}=0$ for all $\bm{\alpha}\in\mathbb{F}_2^{n_2}$. Further, since $g\neq 0$, there
	must be some $\bm{\alpha}$ such that $g_{\bm{\alpha}}\neq 0$. Then $g_{\bm{\alpha}}$ is a non-zero annihilator of $f_{\bm{\alpha}}$ and so
	$\sym{deg}(g_{\bm{\alpha}})\geq \sym{AI}(f_{\bm{\alpha}})$. Clearly, $\sym{deg}(g) \geq \sym{deg}(g_{\bm{\alpha}})$, and so 
	$\sym{deg}(g)\geq \sym{AI}(f_{\bm{\alpha}})$. 

	A similar reasoning shows that if $g$ is a non-zero annihilator of $1\oplus f$, then $\sym{deg}(g)\geq \sym{AI}(f_{\bm{\beta}})$ for
	some $\bm{\beta}\in\mathbb{F}_2^{n_2}$.

	Since $\sym{AI}(f)$ is the minimum of the degrees of all the non-zero annihilators of $f$, and the degrees of all the non-zero annihilators of $1\oplus f$, the
	result follows.
\end{proof}

The lower bound on the algebraic immunity of a direct sum which was obtained in~\cite{DBLP:conf/eurocrypt/MeauxJSC16} and is stated in Proposition~\ref{prop-AI-dsum}
can be seen as a corollary of Proposition~\ref{prop-sub-function}.
The idea is that the sub-functions of $f$ obtained by fixing $X_1,\ldots,X_{n_1}$ to arbitrary values are either $h(Y_1,\allowbreak \ldots,\allowbreak Y_{n_2})$ or
$1\oplus h(Y_1,\allowbreak \ldots,\allowbreak Y_{n_2})$. So from Proposition~\ref{prop-sub-function}, we have $\sym{AI}(f)\geq \sym{AI}(h)$. Similarly, by fixing 
$Y_1,\ldots,Y_{n_2}$ to arbitrary values, we obtain $\sym{AI}(f)\geq \sym{AI}(g)$. 

%

\begin{theorem}\label{thm-ai-lbnd-cons-iter}
	Let $G$, $H$ and $F$ be the functions constructed as in Method~\ref{cons-iter}. Then
	\begin{enumerate}
		\item $\sym{AI}(G), \sym{AI}(H)\geq \min\{\sym{AI}(g),\sym{AI}(h)\}$.
		\item $\sym{AI}(F)\geq \min\{\sym{AI}(G), \sym{AI}(H)\}\geq \min\{\sym{AI}(g),\sym{AI}(h)\}$.
	\end{enumerate}
\end{theorem}
\begin{proof}
	By setting the variables $X_{n+3},X_{n+2},X_{n+1}$ to arbitrary values, the sub-functions of $G$ that are obtained are 
	$g$, $1\oplus g$, $h$, and $1\oplus h$. Since $g$ and $1\oplus g$ have the same algebraic immunity, and $h$ and $1\oplus h$ also have the
	same algebraic immunity, the lower bound on $\sym{AI}(G)$ follows from Proposition~\ref{prop-sub-function}. The lower bounds on 
	$\sym{AI}(H)$ and $\sym{AI}(F)$ follow in a similar manner using Proposition~\ref{prop-sub-function}.
\end{proof}

Method~\ref{cons-iter} can be iterated to obtain functions on progressively larger number of variables. The following construction describes the idea.

\begin{method}\label{cons-iter-grow}
Let $n$ and $t$ be positive integers, and let $g$ and $h$ be two $n$-variable functions. Consider the following iterated construction.
	\begin{tabbing}
		\ \ \ \ \=\ \ \ \ \=\ \ \ \ \kill
		\> $g^{(0)}\leftarrow g$; $h^{(0)}\leftarrow h$; $f^{(0)}\leftarrow \sym{Concat}(g^{(0)},h^{(0)})$; \\
		\> for $i\leftarrow 1$ to $t$ do \\
		\> \> $(g^{(i)},h^{(i)})\leftarrow \sym{Step}(g^{(i-1)},h^{(i-1)})$; $f^{(i)}\leftarrow \sym{Concat}(g^{(i)},h^{(i)})$; \\
		\> end for; \\
		\> return $f^{(t)}$.
	\end{tabbing}
	We denote the function $f^{(t)}$ by $\sym{Iter}_t(g,h)$.
\end{method}

The properties of the function constructed using Method~\ref{cons-iter-grow} are given in the following result.
\begin{theorem}\label{thm-iter-grow}
	Let $n$ and $t$ be positive integers, $g$ and $h$ be $n$-variable functions. Let $f=\sym{Concat}(g,h)$, and $f^{(t)}=\sym{Iter}_t(g,h)$. Then the following holds.
	\begin{compactenum}
	\item The function $f^{(t)}$ is an $(n+3t+1)$-variable function.
	\item If $g$ and $h$ are $m$-resilient, then $f^{(t)}$ is $(m+2t)$-resilient.
	\item $\sym{LB}(f^{(t)})=2^{-t}\cdot \sym{LB}(f)$.
	\item $\sym{AI}(f^{(t)})\geq \min\{\sym{AI}(g),\sym{AI}(h)\}$.
	\item Obtaining $f^{(t)}$ from $g$ and $h$ requires $8t+2$ XOR and $4t+2$ AND gates.
	\end{compactenum}
\end{theorem}
\begin{proof}
	The result follows from Theorem~\ref{thm-iterated-prop} by induction on $t$.
\end{proof}

By appropriately choosing the initial functions $g$ and $h$, Theorem~\ref{thm-iter-grow} can be used to obtain the resiliency, linear bias and lower bound
on algebraic immunity of the function $\sym{Iter}_t(g,h)$. Based on this idea, we define the following two constructions.

\begin{construction}\label{cons-iter-1}
Let $n$ and $t$ be positive integers such that $n-3t-1\geq 3$ is an odd integer. Let $k=(n-3t-2)/2$, and $\psi:\{0,1\}^k\rightarrow\{0,1\}^k$ be a 
bit permutation. Define
        \begin{eqnarray}
                g^{(0)}(X_1,U_1,\ldots,U_k,V_1,\ldots,V_k) & = & X_1 \oplus (\psi,\sym{Maj}_k)\mbox{-}\sym{MM}_{2k}(U_1,\ldots,U_k,V_1,\ldots,V_k) \nonumber \\
                h^{(0)}(X_1,U_1,\ldots,U_k,V_1,\ldots,V_k) & = & 1\oplus X_1 \oplus (\psi,\sym{Maj}_k)\mbox{-}\sym{MM}_{2k}(U_1,\ldots,U_k,V_1,\ldots,V_k) \nonumber \\
                f^{(t)} & = & \sym{Iter}_t(g^{(0)},h^{(0)}). \nonumber
        \end{eqnarray}
\end{construction}

\begin{construction}\label{cons-iter-2}
Let $n$ and $t$ be positive integers such that $n-3t-1\geq 4$ is an even integer. Let $k=(n-3t-5)/2$, and $\psi:\{0,1\}^k\rightarrow\{0,1\}^k$ be a bit 
permutation. Define
                \begin{eqnarray}
                        \lefteqn{g^{(0)}(X_1,Z_1,Z_2,Z_3,U_1,\ldots,U_k,V_1,\ldots,V_k)} \nonumber \\
                        & = & Z_1\oplus Z_2 \oplus X_1(Z_1\oplus Z_3) \oplus (\psi,\sym{Maj}_k)\mbox{-}\sym{MM}_{2k}(U_1,\ldots,U_k,V_1,\ldots,V_k) \nonumber \\
                        \lefteqn{h^{(0)}(X_1,Z_1,Z_2,Z_3,U_1,\ldots,U_k,V_1,\ldots,V_k)} \nonumber \\
                        & = & Z_1\oplus Z_3\oplus X_1Z_2\oplus (\psi,\sym{Maj}_k)\mbox{-}\sym{MM}_{2k}(U_1,\ldots,U_k,V_1,\ldots,V_k) \nonumber \\
                f^{(t)} & = & \sym{Iter}_t(g^{(0)},h^{(0)}). \nonumber
        \end{eqnarray}

\end{construction}
The properties of the functions obtained using Constructions~\ref{cons-iter-1} and~\ref{cons-iter-2} are stated in the following theorem.
\begin{theorem}\label{thm-gen-cons}
Let $n$ and $t$ be positive integers.
	\begin{enumerate}
		\item Suppose $n-3t-1\geq 3$ is an odd integer and $f^{(t)}$ is the function obtained from Construction~\ref{cons-iter-1}.
	Then $f^{(t)}$ is an $n$-variable, $2t$-resilient function with linear bias equal to $2^{-(n-t-2)/2}$, algebraic immunity at least 
			$\lceil (n-3t-2)/4\rceil$, and can be implemented using $4t+k+2$ AND, $8t+k+3$ XOR, and $G(k)$ NAND gates.
\item Suppose $n-3t-1\geq 4$ is an even integer and $f^{(t)}$ is the function obtained from Construction~\ref{cons-iter-2}. 
	Then $f^{(t)}$ is an $n$-variable, $(2t+1)$-resilient function with linear bias equal to $2^{-(n-t-1)/2}$, algebraic immunity at least
			$\lceil (n-3t-5)/4 \rceil$, and can be implemented using $4t+k+4$ AND, $8t+k+8$ XOR, and $G(k)$ NAND gates.
	\end{enumerate}
\end{theorem}
\begin{proof}
	First suppose $n-3t-1$ is odd. 
	
	Note that due to the addition of the variable $X_1$, both $g^{(0)}$ and $h^{(0)}$ are balanced, i.e. 0-resilient.
	From Proposition~\ref{prop-AI-dsum}, the algebraic immunities of both $g^{(0)}$ and $h^{(0)}$ are at least the algebraic immunity of
	$(\psi,\sym{Maj}_k)\mbox{-}\sym{MM}_{2k}(U_1,\allowbreak \ldots,\allowbreak U_k,\allowbreak V_1,\allowbreak \ldots,\allowbreak V_k)$; from Theorem~\ref{thm-MMMaj-AI} 
	the algebraic immunity of $(\psi,\sym{Maj}_k)\mbox{-}\sym{MM}_{2k}(U_1,\allowbreak \ldots,\allowbreak U_k,\allowbreak V_1,\allowbreak \ldots,\allowbreak V_k)$ is at least 
	the algebraic immunity of $\sym{Maj}_k$; and from Theorem~\ref{thm-maj-prop}, the algebraic immunity of $\sym{Maj}_k$ is at least $\lceil k/2\rceil=\lceil (n-3t-2)/4\rceil$. 
	So the algebraic immunities of both $g^{(0)}$ and $h^{(0)}$ are at least $\lceil (n-3t-2)/4\rceil$. Define 
	\begin{eqnarray*}
		\lefteqn{f^{(0)}(X_1,X_2,U_1,\ldots,U_k,V_1,\ldots,V_k)} \nonumber \\
		& = & (1\oplus X_2)g^{(0)}(X_1,U_1,\ldots,U_k,V_1,\ldots,V_k) \oplus X_2h^{(0)}(X_1,U_1,\ldots,U_k,V_1,\ldots,V_k). 
	\end{eqnarray*}	
	Simplifying we have
	$f^{(0)}(X_1,X_2,U_1,\ldots,U_k,V_1,\ldots,V_k)=X_1\oplus X_2\oplus 
	(\psi,\sym{Maj}_k)\mbox{-}\sym{MM}_{2k}(U_1,\allowbreak \ldots,\allowbreak U_k,\allowbreak V_1,\allowbreak \ldots,\allowbreak V_k)$.
	From Theorem~\ref{prop-direct-m>1-even}, the linear bias of $f^{(0)}$ is equal to $2^{-(n-3t-2)/2}$.
	Note that $g^{(0)}$ and $h^{(0)}$ are functions of $n-3t-1$ variables and $f^{(0)}$ is a function of $n-3t$ variables. 
	Since $f^{(t)}=\sym{Iter}_t(g^{(0)},h^{(0)})$, from Theorem~\ref{thm-iter-grow} it follows that $f^{(t)}$ is a function of $n$ variables.
	Further, also from Theorem~\ref{thm-iter-grow}, the resiliency of $f^{(t)}$ is $2t$, linear bias is equal to $2^{-t}\sym{LB}(f^{(0)})=2^{-(n-t-2)/2}$,
	and algebraic immunity is at least $\lceil (n-3t-2)/4\rceil$. 
	The implementation of $f^{(t)}$ requires the implementation of $\sym{Maj}_k$, and the
	implementation of the bit permutation $\psi$. The implementation of the bit permutation $\psi$ does not require any gates 
	(only the appropriate connection
	pattern needs to be implemented). $\sym{Maj}_k$ requires $G(k)$ NAND gates. The counts of the other gates follow directly from 
	Construction~\ref{cons-iter-1} and Theorem~\ref{thm-iter-grow}.

	Next suppose $n-3t-1$ is even. By an argument similar to the case for $n-3t-1$ is odd, the algebraic immunities of both $g^{(0)}$
	and $h^{(0)}$ are at least $\lceil k/2\rceil=\lceil (n-3t-5)/4\rceil$. Define 
	\begin{eqnarray*}	
		\lefteqn{f^{(0)}(X_1,X_2,Z_1,Z_2,Z_3,U_1,\ldots,U_k,V_1,\ldots,V_k)} \nonumber \\
		& = & (1\oplus X_2)g^{(0)}(X_1,Z_1,Z_2,Z_3,U_1,\ldots,U_k,V_1,\ldots,V_k) \nonumber \\
		&   & \oplus X_2h^{(0)}(X_1,Z_1,Z_2,Z_3,U_1,\ldots,U_k,V_1,\ldots,V_k). 
	\end{eqnarray*}
	Simplifying we have
	\begin{eqnarray*}
		\lefteqn{f^{(0)}(X_1,X_2,Z_1,Z_2,Z_3,U_1,\ldots,U_k,V_1,\ldots,V_k)} \nonumber \\
		& = & f_5(X_1,X_2,Z_1,Z_2,Z_3) \oplus (\psi,h)\mbox{-}\sym{MM}_{2k}(U_1,\ldots,U_k,V_1,\ldots,V_k).
	\end{eqnarray*}
	From Theorem~\ref{prop-direct-m>1-odd}, $f^{(0)}$ is $1$-resilient. 
	Note that $Z_1\oplus Z_2\oplus X_1(Z_1\oplus Z_3)=(1\oplus X_1)(Z_1\oplus Z_2)\oplus X_1(Z_2\oplus Z_3)$, which is the concatenation of two
	linear functions both of which are non-degenerate on 2 variables, and hence is 1-resilient. Since $g^{(0)}$ is the direct sum of the 1-resilient function
	$Z_1\oplus Z_2\oplus X_1(Z_1\oplus Z_3)$ and a bent function, 
	it follows that $g^{(0)}$ is 1-resilient.  Similarly, we may write
	$Z_1\oplus Z_3\oplus X_1Z_2=(1\oplus X_1)(Z_1\oplus Z_3)\oplus X_1(Z_1\oplus Z_2\oplus Z_3)$ to see that $h^{(0)}$ is the direct sum
	of a 1-resilient function and a bent function, 
	and so it follows that $h^{(0)}$ is 1-resilient.
	From Theorem~\ref{prop-direct-m>1-odd}, the linear bias of $f^{(0)}$ is equal to $2^{-(2k+6-2)/2}=2^{-(n-3t-1)/2}$.
	As in the case for $n-3t-1$ being odd, from Theorem~\ref{thm-iter-grow}, $f^{(t)}$ is a function of $n$ variables, which is $(2t+1)$-resilient,
	having linear bias equal to $2^{-t}\sym{LB}(f^{(0)})=2^{-(n-t-1)/2}$, algebraic immunity at least $\lceil (n-3t-5)/4\rceil$. 

	As in the previous case, the bit permutation $\psi$ does not require any gate, $\sym{Maj}_k$ requires $G(k)$ NAND gates, and
	the counts of the other gates follow directly from Construction~\ref{cons-iter-2} and Theorem~\ref{thm-iter-grow}.
\end{proof}

\section{Comparison \label{sec-comp}}
In this section, we compare the four constructions, namely Construction~\ref{cons-basic-1} to~\ref{cons-iter-2}. 
All the four constructions make use of the majority function to achieve a guaranteed lower bound on algebraic immunity. Further, for all the
constructions, it is the majority function which requires the maximum number of gates. For comparing among the four constructions, we assume
that all of them achieve the same guaranteed lower bound (say $a_0$) on the algebraic immunity. 
This requires setting $k=2a_0-1$ in all the constructions. Under this setting, the parameters of the different constructions are shown in Table~\ref{tab:others}.

\begin{table}[h]
\centering
\small\setlength{\tabcolsep}{5pt}\renewcommand{\arraystretch}{1.25}
\begin{tabular}{@{}l c c c c@{}}
\toprule
 & Proposition~\ref{prop-direct-m>1-even} & Proposition~\ref{prop-direct-m>1-odd} & Theorem~\ref{thm-gen-cons}(1) & Theorem~\ref{thm-gen-cons}(2)\\
 & $k=2a_0-1$ & $k=2a_0-1$ & $k=2a_0-1$ & $k=2a_0-1$\\
\midrule
$n$ & $m+4a_0-1$ & $m+4a_0+2$ & $3t+4a_0$ & $3t+4a_0+3$\\
resiliency & $m$ & $m$ & $2t$ & $2t+1$\\
	$\lvert\log_2\sym{LB}\rvert$ & $2a_0-1$ & $2a_0+1$ & $t+2a_0-1$ & $t+2a_0+1$\\
	$\sym{AI}$ (lower bound) & $a_0$ & $a_0$ & $a_0$ & $a_0$\\
\midrule
\#AND & $2a_0-1$ & $2a_0+3$ & $4t+2a_0+1$ & $4t+2a_0+3$\\
\#XOR & $2a_0+m$ & $2a_0+m+6$ & $8t+2a_0+2$ & $8t+2a_0+7$\\
\midrule
size $-\,G(2a_0-1)$ & $4a_0+m-1$ & $4a_0+m+9$ & $12t+4a_0+3$ & $12t+4a_0+10$\\
\bottomrule
\end{tabular}
\caption{\label{tab:others}Parameters of Proposition~\ref{prop-direct-m>1-even}, Proposition~\ref{prop-direct-m>1-odd} and 
	Theorem~\ref{thm-gen-cons}(1), \ref{thm-gen-cons}(2), all with $k=2a_0-1$. The number of NAND gates is $G(2a_0-1)$ in all four cases.
	}
\end{table}

One may observe that the trade-offs achieved by Proposition~\ref{prop-direct-m>1-even} and~\ref{prop-direct-m>1-odd} are different from
those achieved by the two cases of Theorem~\ref{thm-gen-cons}. For concreteness let us consider Proposition~\ref{prop-direct-m>1-even}
and Theorem~\ref{thm-gen-cons}(1), the comparison in the other cases being similar. 

Setting $2t=m$, ensures that both Proposition~\ref{prop-direct-m>1-even} and Theorem~\ref{thm-gen-cons}(1) achieve the same resiliency (along with the 
same lower bound on algebraic immunity). Table~\ref{tab-6-6-1} shows the comparison between the parameters, while the column labelled `difference' shows
the difference in the parameters.

\begin{table}[h]
\centering
\small\setlength{\tabcolsep}{6pt}\renewcommand{\arraystretch}{1.25}
\begin{tabular}{@{}l c c c@{}}
\toprule
 & Proposition~\ref{prop-direct-m>1-even} & Theorem~\ref{thm-gen-cons}(1) with $t=m/2$ & difference\\
 & $k=2a_0-1$ & $k=2a_0-1$ & (Thm.~6(1) $-$ Prop.~6)\\
\midrule
$n$ & $m+4a_0-1$ & $\tfrac{3m}{2}+4a_0$ & $\tfrac{m}{2}+1$\\
resiliency & $m$ & $m$ & $0$\\
	$\lvert\log_2\sym{LB}\rvert$ & $2a_0-1$ & $\tfrac{m}{2}+2a_0-1$ & $\tfrac{m}{2}$\\
	$\sym{AI}$ (lower bound) & $a_0$ & $a_0$ & $0$\\
\midrule
\#AND & $2a_0-1$ & $2m+2a_0+1$ & $2m+2$\\
\#XOR & $2a_0+m$ & $4m+2a_0+2$ & $3m+2$\\
\midrule
size $-\,G(2a_0-1)$ & $4a_0+m-1$ & $6m+4a_0+3$ & $5m+4$\\
\bottomrule
\end{tabular}
\caption{\label{tab-6-6-1}Parameters of Proposition~6 and Theorem~6(1) with $2t=m$ and $k=2a_0-1$. The number of NAND gates is $G(2a_0-1)$ in both cases. 
	}
\end{table}

For the same order of resiliency $m$ and the same lower bound $a_0$ on algebraic immunity, Theorem~\ref{thm-gen-cons}(1) has linear bias smaller by a 
factor of $2^{-m/2}$, at the cost of $m/2+1$ more variables and $5m+4$ more gates. For example, for $a_0=30$ and $m=8$ (so $t=4$), 
Proposition~\ref{prop-direct-m>1-even} gives 
$n=127$, linear bias $2^{-59}$, and $127+G(59)$ gates, while Theorem~\ref{thm-gen-cons}(1) gives $n=132$, linear bias $2^{-63}$, and $171+G(59)$ gates.


\section{Conclusion \label{sec-conclu} }
We have described several constructions which provide functions with provable trade-offs between resiliency, linear bias, and algebraic immunity
\textit{and} at the same time are efficient to implement. From a cryptographic perspective, the efficiency of implementation is important for functions
to be actually used in practical designs. We have compared our constructions to the known constructions in the literature and have identified
the trade-offs. We hope our work will stimulate further research in the area of efficiently implementable Boolean functions
with improved trade-offs between the security parameters.

\section*{Acknowledgement} We gratefully acknowledge comments on an earlier version which helped to improve the paper.


\end{document}